\documentclass[pra,10pt,twocolumn,reprint,superscriptaddress,floatfix,showpacs]{revtex4-2}

\usepackage[utf8]{inputenc}
\usepackage[T1]{fontenc}     
\usepackage[british]{babel}  
\usepackage[sc,osf]{mathpazo}
\usepackage[scaled=0.86]{berasans}  
\usepackage[colorlinks = true,
linkcolor = blue,
urlcolor  = blue,
citecolor = blue,
anchorcolor = blue]{hyperref}  
\usepackage{graphicx} 
\usepackage{epstopdf}
\usepackage{subfig}
\usepackage[babel]{microtype}  
\usepackage{amsmath,amssymb,amsthm,bm,amsfonts,mathrsfs,bbm} 

\usepackage{xspace}  
\usepackage{pgfplots}
\usepackage{xcolor,colortbl}
\def\ba{\begin{equation}}
	\def\ea{\end{equation}}
\def\bea{\begin{eqnarray}}
	\def\eea{\end{eqnarray}}
\def\ben{\begin{equation*}}
	\def\een{\end{equation*}}
\def\bean{\begin{eqnarray*}}
	\def\eean{\end{eqnarray*}}
\def\bma{\begin{mathletters}}
	\def\ema{\end{mathletters}}
\def\bi{\begin{itemize}}
	\def\ei{\end{itemize}}

\newcommand{\be}{\begin{equation}}
	\newcommand{\ee}{\end{equation}}

\newcommand{\kommentar}[1]{}

\newcommand{\forget}[1]{}

\newtheorem{theorem}{Theorem}

\newtheorem{definition}{Definition}

\begin{document}
	
	\title{Hidden Non $n$-locality In Linear Networks}
	\author{Kaushiki Mukherjee}
	\email{kaushiki.wbes@gmail.com}
	\affiliation{Department of Mathematics, Government Girls' General Degree College, Ekbalpore, Kolkata-700023, India.}
	\author{Soma Mandal}
	\email{soma2778.wbes@gmail.com}
	\affiliation{Department of Physics, Government Girls' General Degree College, Ekbalpore, Kolkata-700023, India.}
	\author{Tapaswini Patro}
	\email{p20190037@hyderabad.bits-pilani.ac.in}
	\affiliation{Department of Mathematics, Birla Institute of Technology and Science Pilani, Hyderabad Campus,Telangana-500078, India}
	\author{Nirman Ganguly}
	\email{nirmanganguly@hyderabad.bits-pilani.ac.in}
	\affiliation{Department of Mathematics, Birla Institute of Technology and Science Pilani, Hyderabad Campus,Telangana-500078, India}
	
	
	\begin{abstract}
		We study hidden nonlocality in a linear network with independent sources. In the usual paradigm of Bell nonlocality, there are certain states which exhibit nonlocality only after the application of suitable local filtering operations, which in turn are some special stochastic local operations assisted with classical communication (SLOCC). In the present work, we introduce the notion of hidden non $n$-locality. The notion is detailed using a bilocal network. We provide instances of hidden non bilocality and  non trilocality,  where we notice quite intriguingly that non bilocality is observed even when one of the sources distributes a mixed two-qubit separable state. Furthermore a characterization of hidden non bilocality is also provided in terms of the Bloch-Fano decomposition, wherein we conjecture that to witness hidden non bilocality, one of the two states (used by the sources) must have non-null local Bloch vectors. Noise is inevitable in practical scenarios, which makes it imperative to study any possible method to enhance possibility of detecting non-classicality in the presence of noise in the network. We find that local filtering enhances the robustness to noise, which we demonstrate using bit flip and amplitude damping channels.
	\end{abstract}
	\date{\today}
	\maketitle
	
	
	\section{Introduction}\label{intro}
	The study on correlations unachievable within the classical realm has both foundational \cite{Bel} and pragmatic \cite{P.Zoller} implications. Bell nonlocality \cite{Bel,brunrev} constitutes one of the most profound correlations that a quantum state has to offer. The fact that measurements done by spatially separated parties give rise to correlations that cannot be explained by local hidden variables, is the mainstay of such nonlocal correlations \cite{Bel}. Correlations that do not admit a local hidden variable (LHV) description will hence violate a suitably chosen Bell's inequality \cite{Bel}. Thus the violation of Bell's inequality bears the signature of Bell nonlocality. Apart from foundational interest, Bell nonlocality also plays significant roles in practical tasks like device-independent quantum cryptography \cite{A.Acin} and random number generation \cite{random}.
	
	In a standard $(n,m,k)$ measurement scenario, each of $n$ parties sharing a given state repeatedly makes a random and independent choice of one measurement from a collection of $m$ measurements which are each $k$-valued. It is then checked whether the correlations generated therein violate Bell's inequality. Violation of at least one Bell's inequality thus guarantees the nonlocal nature of such correlations. Entanglement is considered a necessity for the violation of Bell's inequalities. However, there are several states, which although entangled, do not violate any Bell's inequality \cite{brunrev,hir1}. Some of those states violate Bell's inequality when subjected to sequential measurements. In such a sequential measurement scenario, the measurements are applied in multiple stages. Initially, the parties are allowed to perform local operations assisted with classical communication (LOCC). In the final step, the parties perform local measurements as in the usual $(n,m,k)$ scenario.
	
	Speaking of sequential measurements, the application of local filtering operations followed by local measurements deserves special mention in the context of Bell nonlocality. Local filtering operations constitute an important class of SLOCC (Stochastic Local Operations and Classical Communication \cite{hir1,hor}). Any state which violates Bell's inequality after being subjected to suitable filtering operations is said to exhibit hidden nonlocality \cite{pop1,gis2}. Over the years, multiple probes have observed various instances of hidden nonlocality \cite{pop1,gis2,hir2,bp1}. In \cite{pop1,gis2}, the authors have given instances of Bell-CHSH \cite{J.F.Clauser} local \cite{brunrev} entangled states which exhibit hidden nonlocality when subjected to suitable local filters. In \cite{hir2}, the authors have shown that even states admitting a LHV model can generate hidden nonlocality under a suitable measurement context. In a broader sense, the present work characterizes hidden nonlocality in the purview of a linear network (which we briefly state below with the details given in section \ref{nlocal1}).
	
	\par In the last decade, the study of nonlocality has been extended beyond the usual paradigm of a Bell scenario to accommodate and analyze network correlations arising in different experimental setups involving multiple independent sources \cite{frtz1,BRA,BRAN,km1,gis1,internet1, lee,internet2}. Network scenarios, characterized by source independence (\textit{$n$-local}) assumption are commonly known as \textit{$n$-local networks} \cite{km1}. In such scenarios, each of the sources sends particles to a subset of distant parties forming the network. Owing to $n$-local assumption, some novel quantum correlations are observed in a network that are not witnessed in the standard Bell scenario \cite{gis1,bilo1}. For example, non-classical correlations (non $n$-local correlations) are generated across the entire network even though all the parties do not share any common past. Moreover, in the measurement scenario associated with a network, some or all the parties perform a fixed measurement. This is also in contrast to the standard Bell scenario, where the random and free choice of inputs by each party is crucial to demonstrate Bell nonlocality.
	
	Different research activities have been conducted which provide for the characterization of quantum correlations in $n$-local networks \cite{BRA,BRAN,km1,km2,bilo2,bilo3,km3,bilo4,km4,km5,bilo5,ejm,bilo6,bilo7,nr4,nr1,nr2,birev,nr3,km7}. Much like the usual Bell nonlocality experiments, violation of an $n$-local inequality indicates the presence of $n$-nonlocal correlations. However, when a particular $n$-local inequality is satisfied, we remain inconclusive. It has been shown that in a network if each source distributes a two-qubit pure entangled state then a violation is observed \cite{bilo2}. The same conclusion does not hold in case the source generates some mixed entangled states. The $n$-local inequality fails to capture nonlocality even though there may be some non $n$-local correlations. Thus, it becomes imperative to probe whether local filtering operations can reveal hidden non $n$-local correlations. The present work addresses this question.
	
	In this work, we introduce the notion of \textit{hidden non $n$-locality}. We analyze the nature of quantum correlations in a $n$-local network where at least one party performs local filtering operations after the distribution of qubits by the sources. For a detailed discussion, we consider the simplest $n$-local network, namely a bilocal network ($n$$=$$2$ \cite{BRAN}). We then characterize the set of hidden nonbilocal correlations. The characterization is also given in terms of the Bloch-Fano decomposition. It is observed that to witness hidden non bilocality in a network, at least one of the two states must have non-null local Bloch vectors, which we state as a conjecture. Interestingly, hidden non bilocality is detected even when one of the sources distributes a two-qubit mixed separable state. Environmental noise is ubiquitous in any implementation of quantum information processing protocols. In this context, it is thus important to explore ways which can enhance the detection of non $n$-local correlations. We find that appropriately chosen local filters are effective in this scenario. We demonstrate this phenomenon using bit flip and amplitude damping channels.

	\par Rest of the work is organized in the following manner:  In sec.\ref{pres}, we briefly discuss the prerequisites for our work. In sec.\ref{ress1}, we have discussed the $n$-local network scenario where now the parties may perform filtering operations thereby introducing the notion of hidden non $n$-locality. Hidden non $n$-local correlations is then analyzed in sec.\ref{ress3}. Characterization of hidden non $n$-locality in terms of Bloch parameters is provided next in sec.\ref{ex3}. Utility of filtering operations in increasing bilocal inequality's robustness to noise is discussed with few examples in sec.\ref{noisy}. We then summarize our work with a discussion on possible future courses of work.
	
	\section{Preliminaries}\label{pres}
	
	\subsection{Bloch-Fano Decomposition of a Density Matrix}
	Let $\rho$ denote an arbitrary two-qubit state. In the Bloch-Fano decomposition $\rho$ is given as:
	\begin{equation}\label{st4}
		\small{\rho}=\small{\frac{1}{4}(\mathbb{I}_{2}\times\mathbb{I}_2+\vec{a}.\vec{\sigma}\otimes \mathbb{I}_2+\mathbb{I}_2\otimes \vec{b}.\vec{\sigma}+\sum_{j_1,j_2=1}^{3}\mathfrak{w}_{j_1j_2}\sigma_{j_1}\otimes\sigma_{j_2})},
	\end{equation}
	where $\vec{\sigma}$$=$$(\sigma_1,\sigma_2,\sigma_3), $ $\sigma_{j_k}$ stand for Pauli operators along three mutually perpendicular directions ($j_k$$=$$1,2,3$). $\vec{a}$$=$$(a_1,a_2,a_3)$ and $\vec{b}$$=$$(b_1,b_2,b_3)$ denote
	local bloch vectors ($\vec{a},\vec{b} $$\in$$\mathbb{R}^3$) corresponding to the party Alice ($A$) and Bob ($B$) respectively with $|\vec{a}|,|\vec{b}|$$\leq$$1$ and $(\mathfrak{w}_{i,j})_{3\times3}$ denotes correlation tensor $\mathcal{W}$ (real).
	Matrix elements $\mathfrak{w}_{j_1j_2}$ are given by $\mathfrak{w}_{j_1j_2}$$=$$\textmd{Tr}[\rho\,\sigma_{j_1}\otimes\sigma_{j_2}].$ \\
	$\mathcal{W}$ can be diagonalized by subjecting it to suitable local unitary operations \cite{gam,luo}. The transformed state is then given by:
	\begin{equation}\label{st41}
		\small{\rho}^{'}=\small{\frac{1}{4}(\mathbb{I}_{2}\times\mathbb{I}_2+\vec{u}.\vec{\sigma}\otimes \mathbb{I}_2+\mathbb{I}_2\otimes \vec{z}.\vec{\sigma}+\sum_{j=1}^{3}s_{j}\sigma_{j}\otimes\sigma_{j})},
	\end{equation}
	$T$$=$$\textmd{diag}(s_{1},s_{2},s_{3})$ denote the correlation matrix in Eq.(\ref{st41}) where $s_{1},s_{2},s_{3}$ are the eigen values of $\sqrt{\mathcal{W}^{T}\mathcal{W}},$ i.e., singular values of $\mathcal{W}.$ It is important to note here such local unitary transforms do not affect the nonlocality exhibited by the state.
	\subsection{Linear $n$-local Networks}\label{nlocal1}
	Here we give a brief overview of linear $n$-local networks \cite{km1}. Let us consider a linear network arrangement of $n$ sources $\textbf{S}_1,\textbf{S}_2,...\textbf{S}_n$ and $n+1$ parties $\textbf{A}_1,\textbf{A}_2,...,\textbf{A}_{n+1}$ (see Fig.\ref{fig1}). $\forall i$$=$$1,2,...,n,$ source $\textbf{S}_i$ independently sends physical systems to $\textbf{A}_i$ and $\textbf{A}_{i+1}.$ Each of $ \textbf{A}_2,\textbf{A}_3,...,\textbf{A}_n $ receives two particles and is referred to as \textit{central} parties. Other two parties $\textbf{A}_1$ and $\textbf{A}_{n+1}$ are referred to as \textit{extreme} parties. Each of the extreme parties receives one particle. Each of the sources $\textbf{S}_i$ is characterized by variable $\lambda_i.$ The sources being independent, joint distribution of the variables $\lambda_1,...,\lambda_n$ is factorizable:
	\begin{equation}\label{tr1}
		q(\lambda_1,...\lambda_n)=\Pi_{i=1}^n q_i(\lambda_i),
	\end{equation}
	where $\forall i,\,q_i$ denotes the normalized distribution of $\lambda_i.$ Eq.(\ref{tr1}) represents the $n$-local constraint. \\
	$\forall i$$=$$2,3,...n-1$ the central party $\textbf{A}_i$ performs a single measurement $y_i$ on the joint state of the two subsystems that are received from $\textbf{S}_{i-1}$ and
	$\textbf{S}_{i}.$
	Each of the two extreme parties ($\textbf{A}_1,$ $\textbf{A}_{n+1}$) selects from a collection of two dichotomous inputs. The $n+1$-
	partite network correlations are local if those can be decomposed as:
	\begin{eqnarray}\label{tr2}
		&&\small{p(o_1,\vec{\mathfrak{o}}_2,...,\vec{\mathfrak{o}}_n,o_{n+1}|y_1,y_2,...,y_n,y_{n+1})}=\nonumber\\
		&&\int_{\Lambda_1}\int_{\Lambda_2}...\int_{\Lambda_n}
		d\lambda_1d\lambda_2...d\lambda_n\,q(\lambda_1,\lambda_2,...\lambda_n) P ,\,\textmd{with}\nonumber\\
		&&P=p(o_1|y_1,\lambda_1)\Pi_{i=2}^n p(\vec{\mathfrak{o}}_i|y_i,\lambda_{i-1},\lambda_i) p(o_{n+1}|y_{n+1},\lambda_n)\nonumber\\
		&&
	\end{eqnarray}
	Notations appearing in Eq.(\ref{tr2}) are detailed below:
	\begin{itemize}
		\item $\forall i,$ $\Lambda_i$ denotes the set of all possible values of $\lambda_i.$
		\item $y_1,y_{n+1}$$\in\{0,1\}$ label inputs of $\textbf{A}_1$ and $\textbf{A}_{n+1}$ respectively.
		\item $o_1,o_{n+1}$$\in$$\{0,1\}$ denote outputs of $\textbf{A}_1$ and $\textbf{A}_{n+1}$ respectively.
		\item $\forall i,$ $\vec{\mathfrak{o}}_i$$=$$(o_{i1},o_{i2})$ labels four outputs of input $y_i$ for $o_{ij}$$\in$$\{0,1\}$
		
	\end{itemize}
	$n+1$-partite correlations are $n$-local if they satisfy both Eqs.(\ref{tr1},\ref{tr2}). Hence, any set of correlations that do not satisfy both Eqs. (\ref{tr1},\ref{tr2}), are termed as non $n$-local.\\
	\par A $n$-local inequality \cite{km1} corresponding to linear $n$-local network is given by:
	\begin{eqnarray}\label{ineqn}
		&& \sqrt{|I|}+\sqrt{|J|}\leq  1,\,  \textmd{where}\nonumber\\
		&& I=\frac{1}{4}\sum_{y_1,y_{n+1}}\langle O_{1,y_1}O_2^0.....O_n^0O_{n+1,y_{n+1}}\rangle\nonumber\\
		&&   J= \frac{1}{4}\sum_{y_1,y_{n+1}}(-1)^{y_1+y_{n+1}}\langle \small{O_{1,y_1}O_2^1...O_n^1O_{n+1,y_{n+1}}}\rangle\,\,\textmd{with} \nonumber\\
		&&   \langle O_{1,y_1}O_2^i.....O_n^iO_{n+1,y_{n+1}}\rangle = \sum_{\mathcal{D}}(-1)^{\textbf{o}_1+\textbf{o}_{n+1}+\textbf{o}_{2i}+...\textbf{o}_{ni}}N_2,\nonumber\\
		&& \textmd{\small{where}}\,N_2=\small{p(\textbf{o}_1,\vec{\mathfrak{o}}_2,...,\vec{\mathfrak{o}}_n,\textbf{o}_{n+1}|y_1,y_{n+1})},\, i=0,1\nonumber\\
		&&\textmd{\small{and}}\, \mathcal{D}=\{\textbf{o}_1,\textbf{o}_{21},\textbf{o}_{22},...,\textbf{o}_{n1},\textbf{o}_{n2},\textbf{o}_{n+1}\}
	\end{eqnarray}
	Violation of Eq.(\ref{ineqn}) guarantees that the corresponding correlations are non $n$-local.
	\begin{center}
		\begin{figure}
			\includegraphics[width=3.5in]{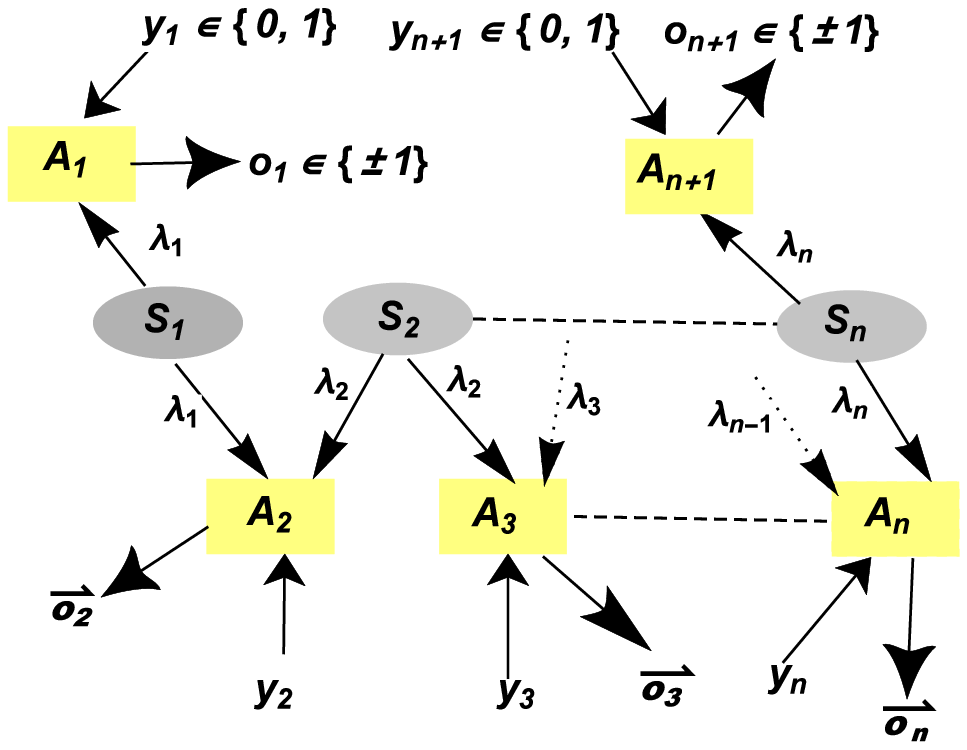} \\
			\caption{\emph{ Schematic diagram of a linear $n$-local network\cite{km1}}}
			\label{fig1}
		\end{figure}
	\end{center}
	\subsection{Quantum Linear $n$-local Network Scenario}\label{pre1}
	In a linear $n$-local network, let $\textbf{S}_i(i$$=$$1,2,...,n)$ generate an arbitrary two qubit state $\varrho_i.$ Each of the central parties thus receives two qubits: one of $\varrho_{i-1}$ and another of $\varrho_i.$ Extreme parties $\textbf{A}_1$ and $\textbf{A}_{n+1}$ receive single qubit of $\varrho_1$ and $\varrho_n$ respectively. Let each of the central parties perform the projective measurement in Bell basis $\{|\psi^{\pm}\rangle,|\phi^{\pm}\rangle\}.$ Let each of $\textbf{A}_1$ and $\textbf{A}_{n+1}$ perform projective measurements along any one of two arbitrary directions. For these measurement settings, non $n$-local correlations are ensured by violation of Eq.(\ref{ineqn}) ,i.e., if \cite{bilo5}:
	\begin{equation}\label{tribd}
		\textbf{B}_{lin}= \sqrt{\Pi_{i=1}^nt_{i1}+\Pi_{i=1}^nt_{i2}}>1
	\end{equation}
	with $t_{i1},t_{i2}$ denoting largest two singular values of correlation tensor ($T_i$) of $\varrho_i\,(i$$=$$1,2,...,n).$ If Eq.(\ref{tribd})
	is not satisfied, nothing can be concluded about the $n$-local nature of corresponding correlations.
	\subsection{Filtering Operations}\label{filter}
	As noted before filtering operations \cite{hir1} are used to reveal hidden nonlocality. Let $\varrho_{AB}$ denote a bipartite state shared between two distant parties Alice and Bob. A local filtering operation by one of the two parties, say, Alice may be defined as a local measurement ($F_A$) having two outcomes $\{\mathfrak{F}_A,\mathfrak{\bar{F}}_A\}$ such that $\mathfrak{F}_A^{\dagger}\mathfrak{F}_A$$+$$\mathfrak{\bar{F}}_A^{\dagger}\mathfrak{\bar{F}}_A$$=$$\mathbb{I}.$ Hence, $\mathfrak{F}_A^{\dagger}\mathfrak{F}_A$$\leq$$\mathbb{I}.$ Local filtering operation $F_B$ can be defined similarly. Let Alice and Bob perform $F_A$ and $F_B$ on their respective subsystems. After the parties apply the filtering operations, $4$ possible output states can be obtained. On being allowed to communicate over the classical channel, Alice and Bob post-select the state corresponding to output pair $(\mathfrak{F}_A,\mathfrak{F}_B),$ i.e., they keep the state:
	\begin{equation}\label{fil1}
		\varrho_{AB}^{'}=\frac{(\mathfrak{F}_A\otimes\mathfrak{F}_B)\varrho_{AB}(\mathfrak{F}_A\otimes
			\mathfrak{F}_B)^{\dagger}}{\textmd{Tr}((\mathfrak{F}_A\otimes\mathfrak{F}_B)\varrho_{AB}
			(\mathfrak{F}_A\otimes\mathfrak{F}_B)^{\dagger})}.
	\end{equation}
	Probability of obtaining $\varrho_{AB}^{'}$ as the output state is given by $\textmd{Tr}(\mathfrak{F}_A\otimes\mathfrak{F}_B\varrho_{AB}
	\mathfrak{F}_A^{\dagger}\otimes\mathfrak{F}_B^{\dagger}).$\\
	The post-selected state $\varrho_{AB}^{'}$, is usually referred to as the \textit{filtered state}.
	As indicated in \cite{hir1,hir2}, for the qubit case the diagonal form of local filters turns out to be most relevant:
	\begin{eqnarray}\label{fil3}
		\mathfrak{F}=\epsilon |0\rangle\langle 0|+|1\rangle\langle 1| , \epsilon \in [0,1]
	\end{eqnarray}
	
	For our purpose, we have used this particular form of local filter.
	
	\section{Sequential Linear $n$-local Network}\label{ress1}
	We now consider an $n$-local linear network where the parties are allowed to perform local filtering operations. The entire network scenario (see Fig.\ref{fig2}) is now divided into two stages: \textit{Preparation Stage} and \textit{Measurement Stage}.\\
	\textit{Preparation Stage:} As in usual linear $n$-local network (sec.\ref{pre1}), let each of $n$ sources $\textbf{S}_i$ distribute a two-qubit quantum state $\rho_{i,i+1}$ between $\textbf{A}_i$ and $\textbf{A}_{i+1}(i$$=$$1,2,...,n).$ Overall state of the particles shared by all the parties across the entire network is thus given by:
	\begin{equation}\label{fil4}
		\rho_{\small{initial}}=\otimes_{i=1}^{n}\rho_{i,i+1}
	\end{equation}
	On receiving the particles, let each of the parties now perform local filtering operations on their respective subsystems. Local filter applied on a single qubit by each of $\textbf{A}_1$ and $\textbf{A}_{n+1}$ is of the form given by Eq.(\ref{fil3}):
	\begin{equation}\label{fil5}
		\mathfrak{F}_{j}=\epsilon_j |0\rangle\langle 0|+|1\rangle\langle 1|,\,j=1,n+1,\,\textmd{and}\,\epsilon_j\in[0,1]
	\end{equation}
	Clearly, in case $\epsilon_j$$=$$1,$ then $\textbf{A}_j$ ($j$$=$$1,n+1$) does not apply any filtering operation.
	Each of the $n-1$ intermediate parties performs local filters on the joint state of the two qubits (received from two sources). Form of the local filter applied by $\textbf{A}_j(j$$=$$2,3,...,n-1)$ is given by:
	\begin{equation}\label{fil6}
		\mathfrak{F}_{j}=\otimes_{i=1}^2(\epsilon_j^{(i)} |0\rangle\langle 0|+|1\rangle\langle 1|),\,\,\epsilon_j^{(i)}\in[0,1]
	\end{equation}
	If $\epsilon_j^{(1)}$$=$$\epsilon_j^{(2)}$$=$$1,$ then $\textbf{A}_j$ ($j$$=$$2,3,...,n$) does not apply any filtering operation.
	
	The filtered state shared across all the parties takes the form:
	\begin{eqnarray}\label{fil7}
		\rho_{\small{filtered}}&=&N(\otimes_{j=1}^{n+1}  \mathfrak{F}_{j})\rho_{\small{initial}}(\otimes_{j=1}^{n+1}  \mathfrak{F}_{j})^{\dagger}\nonumber\\
		\textmd{where}\,N&=&\frac{1}{\textmd{Tr}((\otimes_{j=1}^{n+1}  \mathfrak{F}_{j})\rho_{\small{initial}}(\otimes_{j=1}^{n+1}  \mathfrak{F}_{j})^{\dagger})}
	\end{eqnarray}
	In Eq.(\ref{fil7}), $N$ denotes the probability of obtaining $\rho_{\small{filtered}}.$ To this end, one may note that in the preparation stage, at least one of the $n+1$ parties performs a filtering operation.\\
	\textit{Measurement Stage:} In this stage each of the parties now performs local measurements on their respective share of particles forming the state $\rho_{\small{filtered}}.$ Measurement context is same as in the usual linear $n$-local network scenario (sec.\ref{pre1}). To be precise, each of the central parties $\textbf{A}_2,\textbf{A}_3,...,\textbf{A}_n$  performs projective measurement in Bell basis $\{|\psi^{\pm}\rangle,|\phi^{\pm}\rangle\}.$ $\forall i$$=$$2,3,...,n,$ let $\textbf{B}_i$ denote the Bell state measurement (BSM \cite{BRAN}) of $\textbf{A}_i.$ Let each of $\textbf{A}_1$ and $\textbf{A}_{n+1}$ perform projective measurements $(\textbf{M}_0,\textbf{M}_1)$ and $(\textbf{N}_0,\textbf{N}_{1})$ respectively along any one of two arbitrary directions: $\{\vec{m}_0.\vec{\sigma},\vec{m}_1.\vec{\sigma}\}$ for $\textbf{A}_1$ and $\{\vec{n}_0.\vec{\sigma},\vec{n}_1.\vec{\sigma}\}$ for $\textbf{A}_{n+1}$ with $\vec{m}_0,\vec{m}_1,\vec{n}_0,\vec{n}_1$$\in$$ \mathbb{R}^3.$ Correlations generated due to the local measurements are then used to test a violation of the $n$-local inequality (Eq.(\ref{ineq})).\\
	\par The measurement settings considered here is the same as that considered for usual linear  $n$-local network\cite{bilo5}. For these set-up the $n$-local inequality (Eq.(\ref{ineqn})) takes the form:
	\begin{eqnarray}\label{ineq}
		\frac{1}{2}\sum_{h=0}^1\sqrt{\langle f_h(\textbf{M}_0,\textbf{M}_1,\textbf{N}_0,\textbf{N}_1)\rangle}\leq 1, \hskip 1cm\\
		\textmd{where},
		f_h(\textbf{M}_0,\textbf{M}_1,\textbf{N}_0,\textbf{N}_1)=A(\textbf{M})\otimes_{r=2}^{n-1}\sigma_{2+(-1)^h}\otimes A(\textbf{N})\nonumber\\
		\textmd{with}\,A(\textbf{X})=(\sum_{j=0}^1(-1)^{h.j}\textbf{X}_j),\,h=0,1\nonumber
	\end{eqnarray}
	For rest of the work Eq.(\ref{ineq}) will be referred to as the $n$-local inequality.
	Note that, it is the preparation stage where the scenario considered here differs from that of the usual linear $n$-local network scenario. In the usual scenario, the parties do not perform any operation in this stage. The overall state used in the measurement stage of the usual scenario is thus $\rho_{\small{initial}},$ in contrast to the post-selected state $\rho_{\small{filtered}}$ in the sequential scenario. Such a state is formed due to local operation and classical communication (sec.\ref{filter}) performed by at least one of $n+1$ parties in the preparation stage of the sequential network scenario.\\
	\par Having introduced the sequential linear $n$-local network scenario, we now proceed to characterize the non $n$-locality of the correlations generated therein.
	\begin{center}
		\begin{figure}
			\includegraphics[width=3.5in]{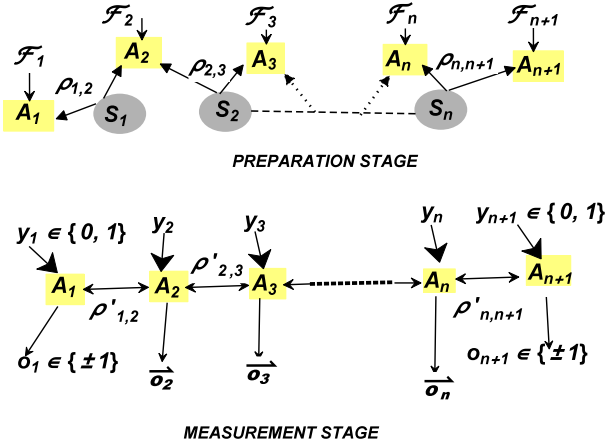} \\
			\caption{\emph{ Schematic diagram of the sequential linear $n$-local network. The overall quantum state shared between the parties in the preparation stage is $\rho_{\small{initial}}$ (Eq.(\ref{fil4})). In this stage, each of the parties performs local filtering operations (Eqs.(\ref{fil5},\ref{fil6})). $\rho_{\small{filtered}}$ (Eq.(\ref{fil7})) is the overall state in the measurement stage.     }}
			\label{fig2}
		\end{figure}
	\end{center}
	\section{Characterization of Hidden non $n$-locality}\label{ress3}
	Consider a sequential linear $n$-local network (Fig.\ref{fig2}) with the filtering operations (Eqs.(\ref{fil5},\ref{fil6})) and measurement context as specified in sec.\ref{ress1}. To be specific each of $\textbf{A}_1$ and $\textbf{A}_{n+1}$ performs projective measurements whereas each of $n-1$ intermediate parties perform Bell basis measurement. Before analyzing hidden non $n$-locality, we first give a formal definition of hidden non $n$-local correlations in such a sequential linear $n$-local network.
	\begin{definition}
		With each of the extreme parties performing projective measurements in anyone of two possible directions and each of the intermediate parties performing a fixed projective measurement in Bell basis, under the $n$-local constraint (Eq.\ref{tr1}), if $n+1$-partite correlations generated in the sequential linear $n$-local network are inexplicable in the form given by Eq.(\ref{tr2}), then such correlations are said to be \textit{hidden non $n$-local correlations} and the corresponding notion of nonlocality is defined as \textit{hidden non $n$-locality}.
	\end{definition}
	
	\par In order to characterize non $n$-locality in sequential network, the term \textit{`hidden'} is used in the same spirit as in \cite{pop1}. Consider a set of $n$ two-qubit states such that non $n$-locality cannot be detected by the violation of $n$-local inequality (Eq.(\ref{ineq})) in the usual $n$-local network. But the same set of states, when used in the sequential $n$-local network, may generate non $n$-local correlations. This corresponds to the detection of hidden non $n$-locality. Violation of the $n$-local inequality Eq.(\ref{ineq}) acts as a sufficient criterion to detect hidden non $n$-local behavior (if any) of the corresponding set of correlations generated in the sequential network scenario.\\
	\par Before progressing further, we would like to note that our entire analysis of non $n$-locality detection will rest upon violation of $n$-local inequality$.\,$As already mentioned in sec.\ref{intro}, violation of such an inequality acts as a sufficient criterion to detect non $n$-locality. It may happen that given a set of $n$ two-qubit states in the $n$-local network, the correlations fail to violate $n$-local inequality. Such correlations may still be non $n$-local. We can rule out non $n$-locality only if we show that the state admits a $n$-local hidden variable model. However, owing to the obvious complexity in giving any such proof, we rather focus on detection issue via the violation of $n$-local inequality. To be precise, when no violation of $n$-local inequality is observed in the usual $n$-local scenario, we use the given set of states in the sequential $n$-local network and test for violation of the same inequality. If the violation is observed, then hidden non $n$-locality is detected. However, one remains inconclusive if no violation is observed.\\
	\par Another important fact to be noted here is that the phenomenon of observing hidden non $n$-locality is stochastic. In a sequential $n$-local network, apart from uncertainty due to measurements (measurement stage), an extra level of uncertainty arises in the preparation stage. As already discussed in sec.\ref{ress1}, such uncertainty is due to the probability $N$ in obtaining the state $\rho_{filtered}$ (Eq.(\ref{fil7})) as the selected output corresponding to the local filtering operations made by the parties. Such a form of uncertainty is absent in the usual non $n$-locality paradigm. Ignoring measurement uncertainty (common in both usual and sequential $n$-local networks), we will refer to the probability term $N$ (Eq.(\ref{fil7})) as the \textit{probability of success for observing hidden non $n$-locality}.
	\par To provide instances of hidden non $n$-locality, we start with the simplest sequential bilocal network.
	\subsection{Examples of Hidden Non bilocality}\label{ex11}
	Let $\textbf{S}_1$ and $\textbf{S}_2$ generate $\varrho_{1,2}$ and $\varrho_{2,3}$ respectively from the following family of two-qubit states \cite{gr1,gr2}:
	\begin{eqnarray}\label{grud1}
		\varrho_{i,i+1}&=&v_i |00\rangle\langle 00|+(1-v_i)(\sin^2x_i|01\rangle\langle 01|+\cos^2x_i|10\rangle\langle 10|\nonumber\\
		&&+\sin x_i\cos x_i(|01\rangle\langle 10|+|10\rangle\langle 01|)),\nonumber\\
		&&i=1,2,\,v_i\in[0,1]\,\textmd{and}\,x_i\in [0,\frac{\pi}{4}]
	\end{eqnarray}
	Let only the intermediate party $\textbf{A}_2$ perform local filtering operations (Eq.(\ref{fil6})) on the joint state of two qubits received from $\textbf{S}_1,\textbf{S}_2.$ For suitable values of local filter parameters ($\epsilon_2^{(1)},\epsilon_2^{(2)}$) and suitable directions of projective measurements $(\vec{m}_0,\vec{m}_1,\vec{n}_0,\vec{n}_1)$ by $\textbf{A}_1$ and $\textbf{A}_{n+1},$ hidden non bilocality is observed (see Fig.\ref{fig3}). For instance, consider two particular states from the above family (Eq.(\ref{grud1})) specified by $(x_1,x_2,v_1,v_2)$$=$$(0.23,0.44,0.1,0.99).$ When used in usual bilocal scenario, L.H.S of Eq.(\ref{tribd}) takes value $0.8871.$ Hence, no violation of the bilocal inequality (Eq.(\ref{ineq}) for $n$$=$$2$) is obtained. But for $(\epsilon_2^{(1)},\epsilon_2^{(2)})$$=$$(0.8,0.97),$ and for suitable measurement settings, L.H.S. of the same inequality (Eq.(\ref{ineq}) gives value $1.081$ with approximately $62\%$ success probability. So the violation reveals hidden non bilocality.
	\begin{center}
		\begin{figure}
			\includegraphics[width=2.6in]{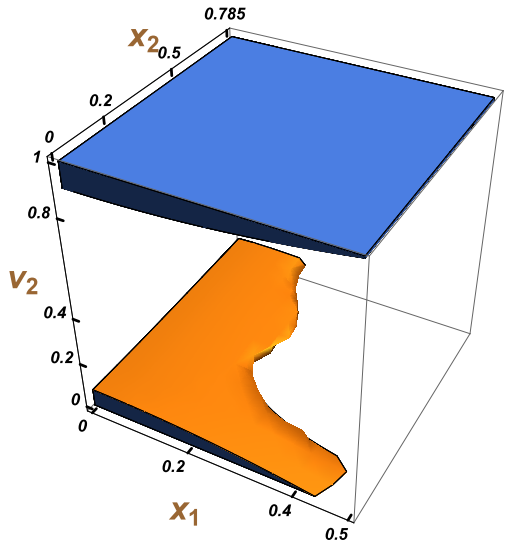} \\
			\caption{\emph{ Shaded region gives state parameters for which hidden non bilocality is observed for $v_1$$=$$0.1$ with not less than $60\%$ probability when only $\textbf{A}_2$ performs local filtering operations (Eq.(\ref{fil6})) for $(\epsilon_2^{(1)},\epsilon_2^{(2)})$$=$$(0.8,0.97).$ }}
			\label{fig3}
		\end{figure}
	\end{center}
	It is interesting to observe that, for states from the same family (Eq.(\ref{grud1})) with $x_1$$=0.23,$ $x_2$$=$$0.34,$ $v_2$$=$$0.15$, hidden non bilocality cannot be detected if only $\textbf{A}_2$ applies local filters. When such states are used in the network, hidden non bilocality can be detected only if all the three parties apply suitable local filters (see Fig.\ref{fig4}).
	\begin{center}
		\begin{figure}
			\includegraphics[width=2.6in]{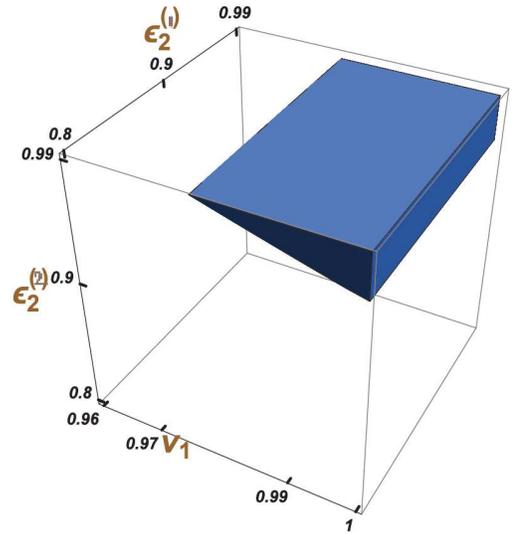} \\
			\caption{\emph{ Let us consider specific members from the family (Eq.(\ref{grud1})):$x_1$$=0.23,$ $x_2$$=$$0.34,$ $v_2$$=$$0.15.$ Shaded region gives state parameter $v_1$ and parameters of local filters $(\epsilon_2^{(1)},\epsilon_2^{(2)})$ applied by $\textbf{A}_2$ for which hidden non bilocality is observed with not less than $30\%$ probability approximately when all the parties apply local filters with extreme parties performing specific local filters for $(\epsilon_1,\epsilon_4)$$=$$(0.95,0.76).$ It may be noted that non bilocality cannot be detected if these states are used in the usual bilocal network.}}
			\label{fig4}
		\end{figure}
	\end{center}
	\subsection{Examples of Hidden Non trilocality}\label{ex1}
	Let us now consider a trilocal sequential network. Let each of $\textbf{S}_1,\textbf{S}_2,\textbf{S}_3$ distribute states from the above family of states (Eq.(\ref{grud1})). Let each of the two intermediate parties $\textbf{A}_2,\textbf{A}_3$ perform local filtering operations on their respective share of particles whereas the extreme parties do not perform any filtering operation. Hidden nontrilocality is observed in the network (see Fig.\ref{fig5}). For example, consider specific state parameters:$(x_1,x_2,x_3,v_1,v_2,v_3)$$=$$(0.3455,0.5586, 0.7799$\\$,0.1,0.12,0.1).$ Non trilocality is not detected when these three states are used in the usual trilocal network (L.H.S of Eq.(\ref{tribd}) takes value $0.9888$). But, under suitable measurement settings and specific filtering parameters: $(\epsilon_2^{(1)},\epsilon_2^{(2)},\epsilon_3^{(1)},\epsilon_3^{(2)})$$=$$(0.6362,0.99,0.989,0.989),$ L.H.S of Eq.(\ref{ineq}) gives $1.2332$ with approximately $44\%$ success probability.
	\begin{center}
		\begin{figure}
			\includegraphics[width=2.6in]{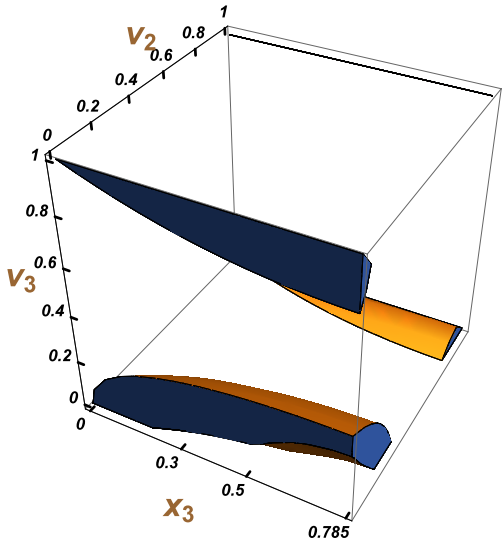} \\
			\caption{\emph{ For specific values of state parameters $(x_1,x_2,v_1)$$=$$(0.3455,0.5586,0.1),$ and specified local filters: $(\epsilon_2^{(1)},\epsilon_2^{(2)},\epsilon_3^{(1)},\epsilon_3^{(2)})$$=$$(0.6362,0.99,0.989,0.989),$ shaded region gives state parameters $x_3,\,v_2,\,v_3$ and parameters for which hidden non trilocality is observed with not less than $44\%$ probability. In case these states are used in the usual trilocal network, non trilocality cannot be detected.}}
			\label{fig5}
		\end{figure}
	\end{center}
	\subsection{Entanglement and Hidden non $n$-locality}\label{ex2}
	If mixed states are allowed in the network, all the sources need not distribute entangled states. For example, let us consider a sequential bilocal network. Let $\textbf{S}_1$ generate mixed entangled state $\varrho_{1,2}$ from the family of states given by Eq.(\ref{grud1}).
	Let $\textbf{S}_2$ distribute separable Werner state \cite{pop1,brunrev}:
	\begin{eqnarray}\label{wer1}
		\varrho_{2,3}&=&\frac{(1-p_2)}{4}\mathbb{I}_{4\times 4}+p_2(|01\rangle\langle 01|+|10\rangle\langle 10|\nonumber\\
		&&-(|01\rangle\langle 10|+|10\rangle\langle 01|))\ , p_2\in[0.25,0.30]
	\end{eqnarray}
	When $\textbf{A}_2$ applies suitable local filtering operation in the preparation stage, then hidden non bilocality is detected for suitable local measurement settings applied in the measurement stage of the network (see subfig.a in Fig.\ref{fig6}). But no violation of bilocal inequality (Eq.(\ref{ineq}) for $n$$=$$2$) is observed when the same states are used in the usual bilocal network \cite{bilo2}.   \\
	\par Let us now consider a sequential trilocal network. Let $\textbf{S}_1,\textbf{S}_3$ each generate a mixed entangled state from the same family of states (Eq.(\ref{grud1})) whereas $\textbf{S}_2$ generate a separable Werner state (Eq.(\ref{wer1})). Under a suitable measurement context, hidden non trilocality can be observed in the network (see subfig.b in Fig.\ref{fig6}). This was noted in \cite{birev} while considering usual network nonlocality, however we observe this phenomenon too in the pursuit to reveal hidden non $n$-locality.\\
	\par All these instances imply that not all the sources need to generate entanglement in a sequential $n$-local network for detecting hidden non $n$-locality. 
	\begin{center}
		\begin{figure}
			\begin{tabular}{c}
				\subfloat[ ]{\includegraphics[trim = 0mm 0mm 0mm 0mm,clip,scale=0.7]{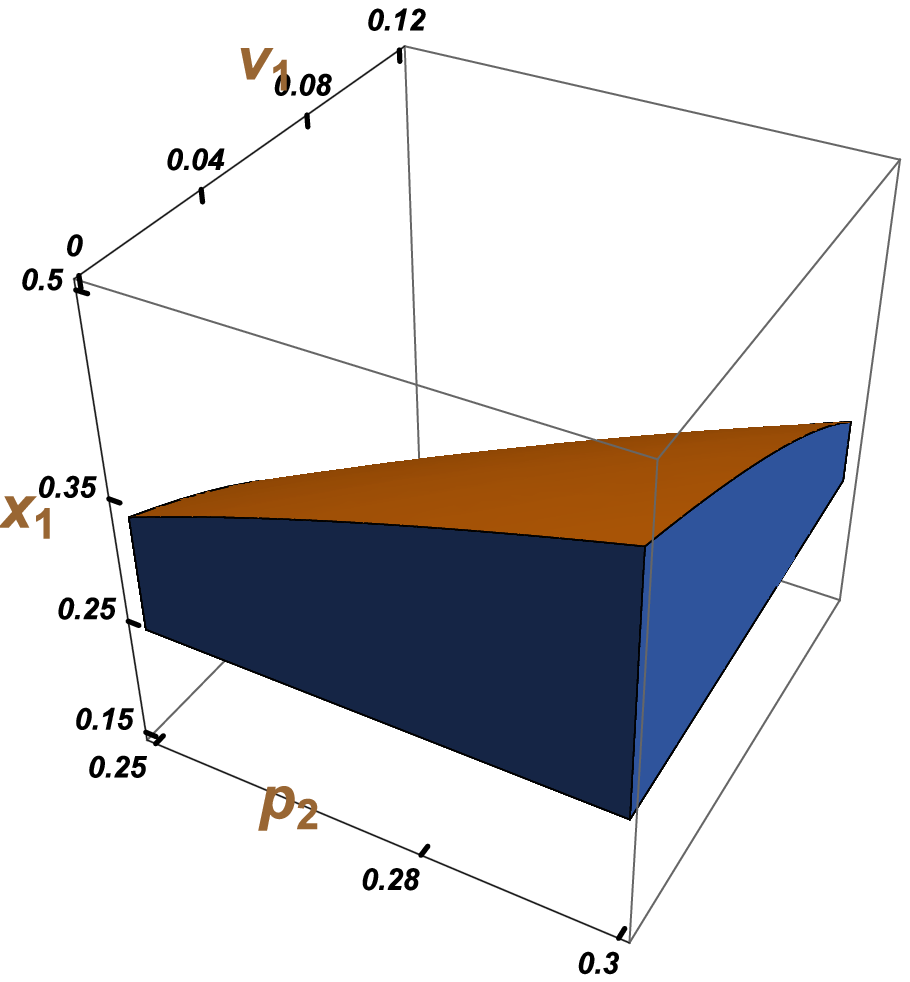}}\\
				\subfloat[]{\includegraphics[trim = 0mm 0mm 0mm 0mm,clip,scale=0.7]{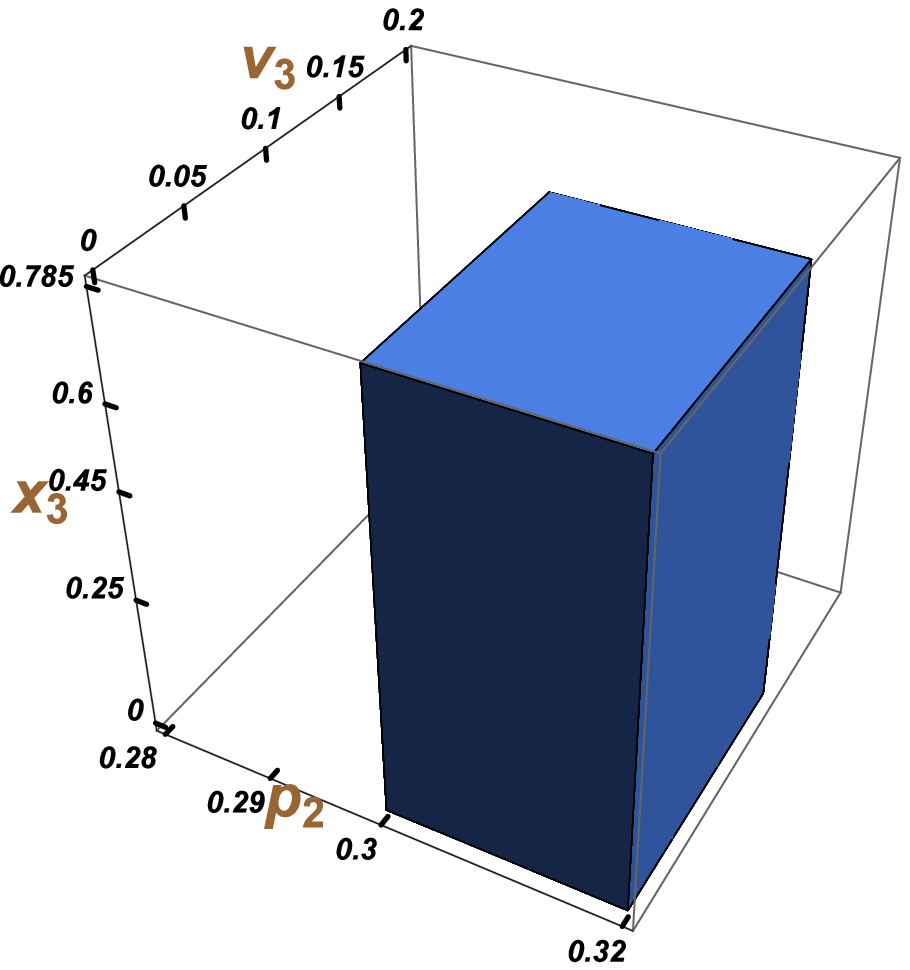}}\\
			\end{tabular}
			\caption{\emph{  In both subfigures, shaded portions indicate regions in the parameter space ($(v_1,x_1,p_2)$ in subfig.a and $(v_3,x_3,p_2)$ in subfig.b) for which hidden non $n$-locality is observed for $n$$=$$2$ (subfig.a) and $n$$=$$3$ (subfig.b). Specifications used in subfig.a are $(\epsilon_2^{(1)}, \epsilon_2^{(2)})$$=$$(0.46,1)$. Specifications used in subfig.b are $(\epsilon_2^{(1)}, \epsilon_2^{(2)}, \epsilon_3^{(1)}, \epsilon_3^{(2)},x_1,v_1)$$=$$(0.762,0.038,0.038,1,0.3,0.07)$. In each of these two cases, violation of $n$-local inequality (Eq.(\ref{ineq}) for $n$$=$$2,3$) is not observed in the usual $n$-local network.}}
			\label{fig6}
		\end{figure}
	\end{center}
	\subsubsection{Comparison with observations in \cite{bilo5}}
	As already discussed in sec.\ref{ress1}, the measurement settings considered here are the same as that considered in usual linear $n$-local network \cite{bilo5}. As per the arguments presented in \cite{bilo5}, non bilocality (non trilocality) can be observed when one (two) maximally entangled state(s) are used along with any generic two qubit state which does not exhibit Bell-CHSH violation. However, in our sequential network scenario, non bilocal (non trilocal) correlations are generated even when one (two) mixed, hence non maximally entangled states are used with a separable state. Consequently, our results, as discussed above in subsec.\ref{ex2}, clearly points out the utility of applying suitable filtering operations in the context of simulating non $n$-locality in linear networks.
	\subsection{No violation of Eq.(\ref{ineq}) while using product of mixed states}
	Next, let us consider the case when only product of two single-qubit mixed states are used in the sequential $n$-local network. Let each party be allowed to perform local filters as mentioned in sec.\ref{ress1}. Hidden non $n$-locality cannot be detected if at least one of the sources distributes a product of two single-qubit mixed states. The result is formalized as follows:
	\begin{theorem}
		In a sequential $n$-local network, with each extreme party performing projective measurements and each of the intermediate party measuring in Bell basis, for any $i$$\in$$\{1,2,...,n\},$ if $i^{th}$ source generates a product of two arbitrary single-qubit mixed states and if the parties perform local filters of the form given by Eqs.(\ref{fil5},\ref{fil6}), then a violation of $n$-local inequality (Eq.(\ref{ineq})) is impossible for any finite $n.$
	\end{theorem}
	\begin{proof}
		See Appendix.
	\end{proof}
	\section{Characterization in terms of Bloch parameters}\label{ex3}
	Here we intend to analyze hidden non $n$-locality detection from density matrix formalism of the states used in the corresponding network$.\,$Examples of hidden non bilocality and non trilocality illustrated in subsecs.\ref{ex11},\ref{ex1}, involve members from a particular family of two-qubit states (Eq.(\ref{grud1})). Now it may be noted that any member $\varrho_i$ from this family has non-null local Bloch vectors:\\
	\begin{eqnarray*}
		u_i &=&(0,0,v_i - (1 - v_i)\cos(2x_i)) \\
		z_i &=&(0,0,v_i +(1 - v_i)\cos(2x_i))
	\end{eqnarray*}
	Again, as discussed in subsec.\ref{ex2}, hidden non $n$-locality (for $n$$=$$2,3$) is observed when one of the states is Werner state. It may be noted that Werner state does not have any local Bloch vector. Combining these two observations from subsecs.\ref{ex11},\ref{ex1} and \ref{ex2}, it is clear that hidden non $n$-locality can be observed when at least one of the states used in the corresponding network has local Bloch vector. At this junction, we conjecture that hidden non $n$-locality cannot be detected via the violation of Eq.(\ref{ineq}) when none of the states used in the network has local Bloch vectors (see Appendix).
	\subsection{Closed Form of Upper Bound of Eq.(\ref{ineq})}
	\par In absence of filtering operations, there exists a closed form in terms of state parameters ($\textbf{B}_{lin}$ in Eq.(\ref{tribd})) of the upper bound of linear $n$-local inequality (Eq.(\ref{ineq})). Following the method discussed in Appendix, the upper bound of Eq.(\ref{ineq}) in linear sequential network scenario maintains the same structure as that of $\textbf{B}_{lin}$ in Eq.(\ref{tribd}):
	\begin{equation}\label{app7}
		\textbf{B}_{seq}= \sqrt{\Pi_{j=1}^nt^{''}_{j1}+\Pi_{j=1}^nt^{''}_{j2}},
	\end{equation}
	where $\forall j$$=$$1,2,...,n,$ the two largest singular values $t^{''}_{1},t^{''}_{2}$ of the normalized post selected states $\rho_{j,j+1}^{''}$ (Eq.(\ref{extra})) are functions of the filtering parameters $\epsilon_1,\epsilon_{n+1},\epsilon^{(1)}_k,\epsilon^{(2)}_k(k$$=$$2,3,...,n-1)$, the singular values of the correlation tensor and also local bloch vectors of $\rho_{j,j+1}.$ So, unlike $\textbf{B}_{lin},$ the closed form of upper bound of the $n$-local inequality (Eq.(\ref{ineq})) depends on state parameters and also on the filtering parameters. Eq.(\ref{ineq}) is thus violated if:
	\begin{equation}\label{jhml}
		\textbf{B}_{seq}>1
	\end{equation}
	For any given set of initial states $\rho_{j,j+1}(j$$=$$1,2,...,n),$ hidden non $n$-local correlations can be simulated in case there exist suitable filters such that Eq.(\ref{tribd}) is violated but Eq.(\ref{jhml}) is satisfied.
	\subsection{Illustration}
	Let us consider the following family of two qubit states\cite{xs1}:
	\begin{eqnarray}\label{ilus1}
		\chi(x_1,x_2,x_3,x_4) = x_1|00\rangle\langle 00|+x_2|01\rangle\langle 01|+ x_3|11\rangle\langle 11|\nonumber\\
		+x_4 (|00\rangle\langle 11|+|11\rangle\langle 00|),\textmd{with}\,x_1+x_2+x_3=1\nonumber\\
		0\leq x_1,x_2,x_3\leq 1,\,x_4^2\leq x_1.x_3\nonumber
	\end{eqnarray}
	This family forms a subclass of X state\cite{xs1}.
	Singular values of the correlation matrix (Eq.(\ref{st41})) of this class are:
	\begin{eqnarray}\label{ilus2}
		s_1 &=& 2|x_4| \nonumber\\
		s_2&=& s_1\nonumber\\
		s_3 &=& |x_1 - x_2 + x_3|
	\end{eqnarray}
	
	Local Bloch vectors are given by:
	\begin{eqnarray}\label{ilus3}
		\vec{u} &=& (0,0,x_1+x_2-x_3)\nonumber \\
		\vec{z} &=& (0,0,x_1-x_2-x_3)
	\end{eqnarray}
	Let any two states from this class of states (Eq.(\ref{ilus1})) be used in the usual bilocal network: $\rho_{i,i+1}$$=$$\chi(x_{i1},x_{i2},x_{i3},x_{i4}),\,i$$=$$1,2.$ For that network, bilocal inequality is not violated if:$\textbf{B}_{lin}$ then turns out to be:
	\begin{eqnarray}\label{ilus4}
		\textbf{B}_{lin}&\leq& 1\,\textmd{\small{where}}\\
		\textbf{B}_{lin}&=&\textmd{Max}[  \sqrt{2\small{L}_1},\sqrt{\small{L}_1+\small{L}_2},\sqrt{\small{L}_1+\small{L}_3},\sqrt{\small{L}_2+\small{L}_3},\nonumber\\
		&&\sqrt{\small{L}_1+\small{L}_4}],\,\textmd{\small{with}}\nonumber\\
		\small{L}_1&=&4|x_{14}.x_{24}|\nonumber\\
		\small{L}_2&=& 2|(x_{11} -  x_{12} +  x_{13})x_{24}|\nonumber\\
		\small{L}_3&=&2|(x_{21} -  x_{22} +  x_{23})x_{14}|\nonumber\\
		\small{L}_4&=& \Pi_{i=1}^2|(x_{i1} -  x_{i2} +  x_{i3})|\nonumber\\
	\end{eqnarray}
	Now let $\chi(x_{i1},x_{i2},x_{i3},x_{i4})(i$$=$$1,2)$ be used in a sequential bilocal network where all the parties are applying the same filtering operation, i.e., $\epsilon_1$$=$$\epsilon_2^{(1)}$$=$$\epsilon_2^{(2)}$$=$$\epsilon_3$$=$$\epsilon$. After the preparation stage, correlation tensor of the normalized post selected states $\rho_{i,i+1}^{''}$ (Eq.(\ref{extra})) turns out to be:
	\begin{eqnarray}\label{ilus5}
		s_{i1}^{''} &=& \frac{2|x_{i4}|\epsilon^2}{|x_{i3} + \epsilon^2 (x_{i2} +x_{i1} \epsilon^2)|} \nonumber\\
		s_{i2}^{''} &=& s_{i1}^{''}\nonumber\\
		s_{i3}^{''} &=&\frac{|x_{i3} + \epsilon^2 (-x_{i2} +x_{i1} \epsilon^2)|}{|x_{i3} + \epsilon^2 (x_{i2} +x_{i1} \epsilon^2)|}
	\end{eqnarray}
	The probability of success is given by $\Pi_{i=1}^2(x_{i3} + \epsilon^2 (x_{i2} +x_{i1} \epsilon^2)).$\\
	As discussed above, $\textbf{B}_{seq}$ is obtained by using the largest two of these singular values (Eq.(\ref{ilus5})) of each of $\rho_{1,2}^{''}$ and $\rho_{2,3}^{''}.$ Violation of the bilocal inequality(Eq.(\ref{ineq}) for $n$$=$$2$) is obtained if:
	\begin{eqnarray}\label{ilus6}
		\textbf{B}_{seq}&>&1,\,\textmd{\small{where}}\\
		\textbf{B}_{seq}&=&\frac{1}{\sqrt{\small{L}_8}}\textmd{\small{Max}}[
		\sqrt{2\small{L}_5},
		\sqrt{\small{L}_5+\small{L}_6},
		\sqrt{\small{L}_5+\small{L}_7},
		\sqrt{\small{L}_6+\small{L}_7},\nonumber\\
		&&\sqrt{\small{L}_5  +\frac{\small{L}_6.\small{L}_7}{4x_{14}x_{24}\epsilon^4}}]
		,\,\textmd{\small{with}}\nonumber\\
		\small{L}_5&=&4|x_{14}x_{24}\epsilon^4| \nonumber \\
		\small{L}_6&=&2|x_{24}\epsilon^2(x_{13}+\epsilon^2(-x_{12}+x_{11}\epsilon^2))| \nonumber\\
		\small{L}_7 &=&2| x_{14}\epsilon^2(x_{23}+\epsilon^2(-x_{22}+x_{21}\epsilon^2))| \nonumber\\
		\small{L}_8&=&\Pi_{i=1}^2(x_{i3}+\epsilon^2(x_{i2}+x_{i1}\epsilon^2))
	\end{eqnarray}
	For suitable value of the filtering parameter $\epsilon,$ there exist states from this subclass of $X$ states (Eq.(\ref{ilus1})) which satisfy both the above relations (Eqs.(\ref{ilus4},\ref{ilus6})) for which hidden nonbilocal correlations are simulated in sequential bilocal network but non bilocality cannot be detected in usual bilocal network (see Fig.\ref{newfig1}). For a numeric instance, let us consider $\chi(0.2,0.1,0.7,0.15)$ and $\chi(0.86,0,0.14,0.33).$ For these two states $\textbf{B}_{lin}$$=$$0.999.$ Hence, no violation of Eq.(\ref{ineq}) observed in usual bilocal network. However, in sequential bilocal network, when all the parties perform filtering with $\epsilon$$=$$0.77,$ violation of the same is observed (with approximately $37\%$ success probability) as $\textbf{B}_{seq}$$=$$1.023.$
	\begin{center}
		\begin{figure}
			\includegraphics[width=2.6in]{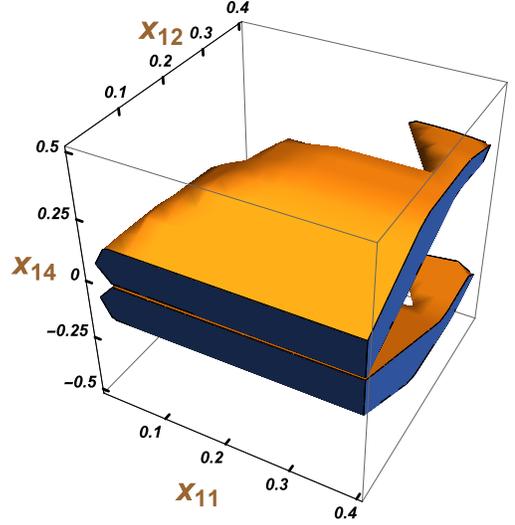} \\
			\caption{\emph{ Shaded region gives the specify the state parameters of $\chi(x_{11},x_{12},x_{13},x_{14})$  (Eq.(\ref{ilus1})), which when used with $\chi(0.2,0.1,0.7,0.15)$ in the sequential bilocal network, simulate hidden non bilocal correlations in case each of the four parties perform apply same local filters specified by $\epsilon$$=$$0.75.$ It may be noted that non bilocality cannot be detected when $\chi(0.2,0.1,0.7,0.15)$ and $\chi(x_{11},x_{12},x_{13},x_{14})$ corresponding to any point in the shaded region, are used in usual bilocal network.}}
			\label{newfig1}
		\end{figure}
	\end{center}
	
	\section{Enhancement In Robustness To noise}\label{noisy}
	A linear $n$-local network scenario underlies different entanglement distribution protocols involving quantum repeaters \cite{birev}. In an idealistic situation pure entangled states are supposed to be communicated among the distant observers (often referred to as \textit{nodes} \cite{reps}) in any such network structure. However, in practical situations, due to unavoidable interaction with the environment, entanglement is transferred across noisy channels \cite{nie}. It is thus significant to study any possible method to enhance possibility of detecting non classicality in the presence of noise in the network. Hence, from practical perspectives, it becomes interesting to explore procedures that can increase resistance to noise of the $n$-local inequality for detecting non $n$-local correlations. Applying suitable local filters turns out to be effective in this context. To be specific it is observed that non bilocality can be detected over a wider range of noise parameter in sequential linear bilocal network in comparison with usual bilocal network. In support of our claim, we provide with illustration considering communication over two specific noisy channels.\\
	\subsection{Communication Through Bit Flip Channel\cite{nie}}
	Let each of the two sources  $\textbf{S}_1,\textbf{S}_2$ generate a pure entangled state:
	\begin{equation}\label{pure1}
		|\Psi(\theta)\rangle=\cos\theta|01\rangle+\sin\theta |10\rangle,\,\theta\in(0,\frac{\pi}{4})
	\end{equation}
	Let each of the two qubits generated from $\textbf{S}_i$ be passed through a bit flip channel parameterized by $p_i(i$$=$$1,2).$ $\forall i$$=$$1,2,\,p_i$ denotes the probability with which the state of a single qubit is flipped from $|0\rangle$ to $|1\rangle$  and vice-versa. Each of $\rho_{1,2}$ and $\rho_{2,3}$ is thus a two qubit mixed entangled state:
	\begin{eqnarray}\label{pure2}
		\rho_{i,i+1}&=&p_i(1-p_i)(|00\rangle\langle 00|+\sin2\theta( |00\rangle\langle 11|+|11\rangle\langle 00|))\nonumber\\
		&&+((1-p_i)^2\cos^2\theta+p_i^2\sin^2\theta) |01\rangle\langle 01|+\nonumber\\
		&&((1-p_i)^2\sin^2\theta+p_i^2\cos^2\theta) |10\rangle\langle 10|+\nonumber\\
		&&(1-2p_i+2p_i^2) \cos\theta\sin\theta(|10\rangle\langle 01|+ |01\rangle\langle 10|)\nonumber\\
		&&i=1,2
	\end{eqnarray}
	On application of suitable local filters by the parties, hidden nonbilocal correlations are simulated over an enhanced range of noise parameters ($p_1,p_2$) compared to the range of $(p_1,p_2)$ for which nonbilocality is detected in usual bilocal network (see Fig.\ref{fignew2}). For a specific instance,  consider two identical copies of $|\Psi(0.62)\rangle$ (Eq.(\ref{pure1})) As discussed above, let these states be passed through bit flip channels, $\rho_{1,2}$ and $\rho_{2,3}$ (Eq.(\ref{pure2})) being characterized by $p_1$$\in$$(0,0.4)$ and $p_2$$=$$0.15$ respectively. In case the parties do not apply filtering, $\textbf{B}_{lin}$$>$$1$ if $p_1$$\in$$(0,0.214].$ Now, in the sequential bilocal network, when $\textbf{A}_2$ applies filtering operations specified by $\epsilon_2^{(1)}$$=$$0.98,$ $\epsilon_2^{(2)}$$=$$0.79,$ $\textbf{B}_{seq}$$>$$1$ is obtained for $p_1$$\in$$(0,0.235].$ Hence, for this particular instance, when $p_1$$\in$$[0.215,0.235]$, nonbilocality can be detected in sequential network but not in usual bilocal network.
	\subsection{Communication Through Amplitude Damping Channel\cite{nie}}
	Let each of $\textbf{S}_1,\textbf{S}_2$  generate an identical copy of the pure entangled state $|\Psi(\theta)\rangle$ (Eq.(\ref{pure1})). Each of the qubits of $|\Psi(\theta\rangle),$ generated from $\textbf{S}_i,$  are distributed through identical amplitude damping channels characterized by damping parameter $\gamma_i(i$$=$$1,2).$ The mixed entangled states thus distributed in the network are given by:
	\begin{eqnarray}\label{pure3}
		\rho_{i,i+1}&=&\gamma_i(|00\rangle\langle 00|+(1-\gamma_i)(\cos^2\theta(|01\rangle\langle 01|+\nonumber\\
		&&\cos\theta\sin\theta ((|10\rangle\langle 01|+(|01\rangle\langle 10|)\nonumber\\
		&&\sin^2\theta(|10\rangle\langle 10|)\,\,i=1,2.
	\end{eqnarray}
	It is observed that there exists range of damping parameters $(\gamma_1,\gamma_2)$ for which hidden nonbilocality can be exploited under effect of suitable filtering operations in contrast to usual bilocal network where no violation of bilocal inequality (Eq.(\ref{ineq})) can be observed (see Fig.\ref{fignew2}). \\
	For example, consider the followings: $|\Psi(0.55)\rangle$ and amplitude damping channels with $\gamma_1$$=$$0.21$ and $\gamma_2$$\in$$(0,1).$ With these specifications, $\textbf{B}_{lin}$$>$$1$ for $\gamma_2$$\in$$(0,0.2].$ Now, for $\epsilon_1$$=0.78,$ $\epsilon_3$$=0.79,$ $\epsilon_2^{(1)}$$=0.22$ and $\epsilon_2^{(2)}$$=0.1,$ $\textbf{B}_{seq}$$>$$1$ for $\gamma_2$$\in$$(0,0.54].$ $(0.2,0.54]$ thus turns out to be the enhanced range of visibility for detecting violation of the bilocal inequality under effective filtering operations.
	\begin{center}
		\begin{figure}
			\begin{tabular}{cc}
				\subfloat[]{\includegraphics[trim = 0mm 0mm 0mm 0mm,clip,scale=0.45]{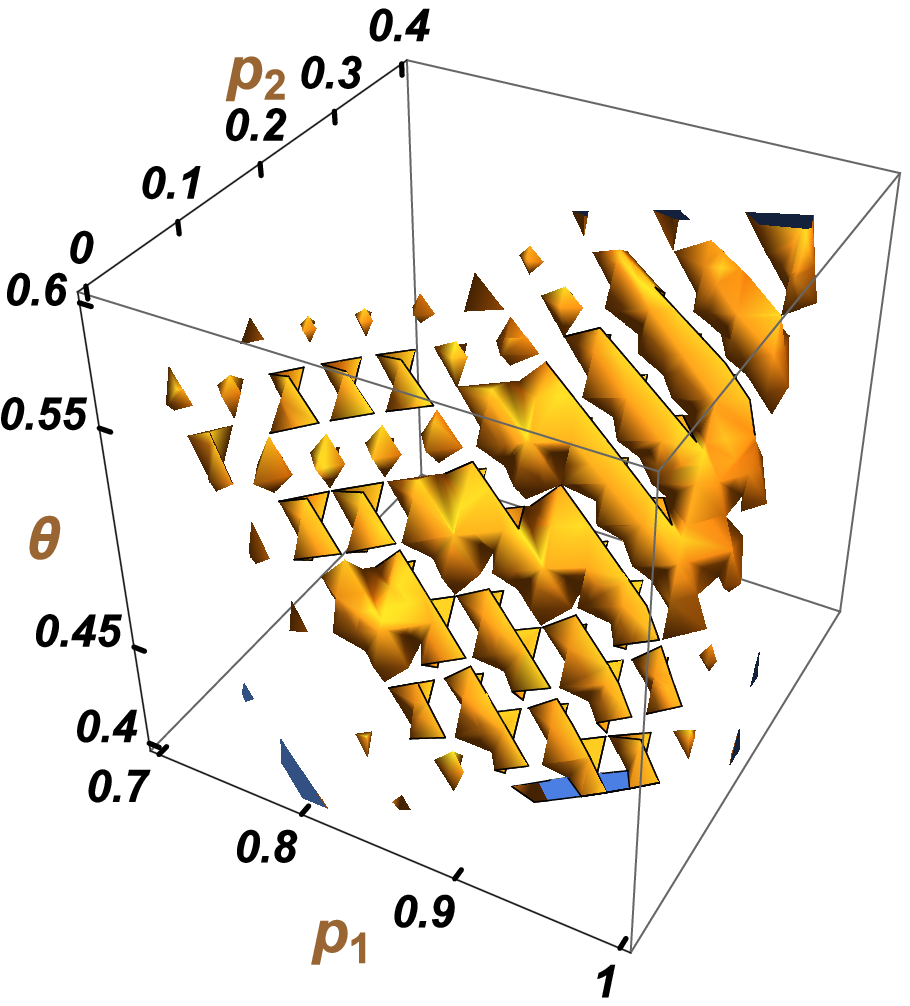}}&\subfloat[]{\includegraphics[trim = 0mm 0mm 0mm 0mm,clip,scale=0.45]{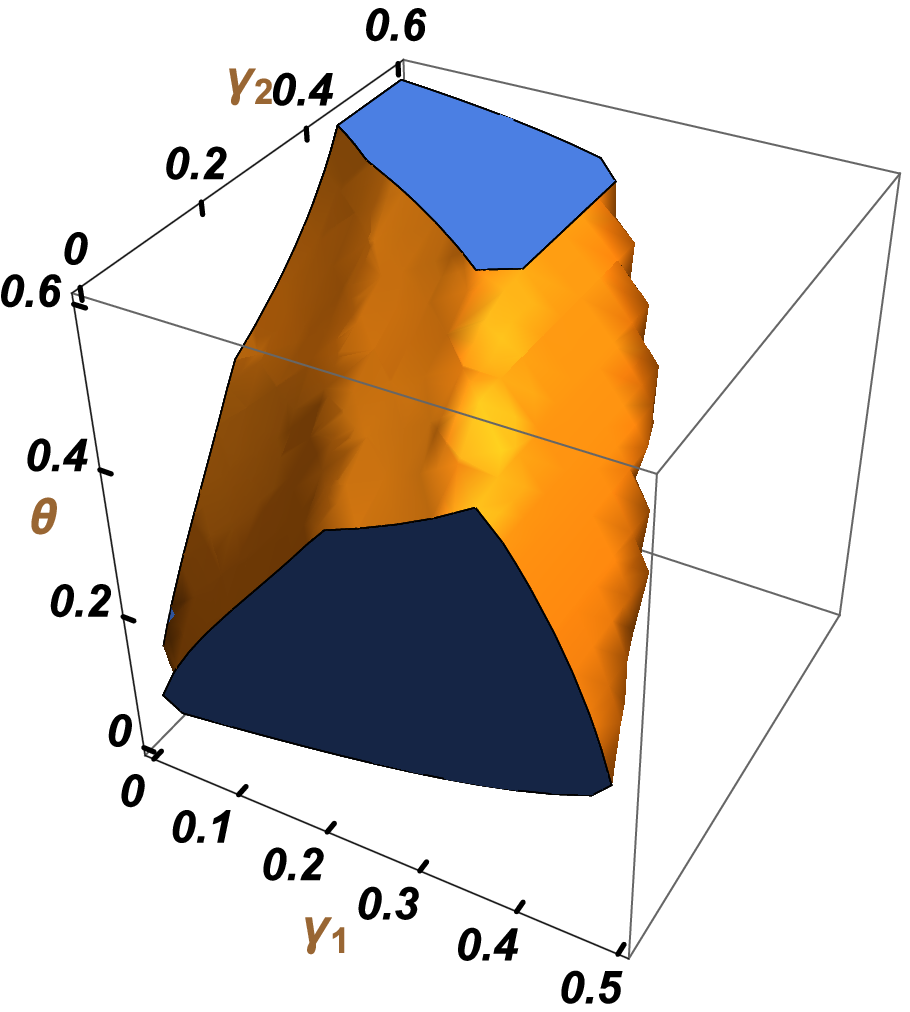}}\\
				\subfloat[]{\includegraphics[trim = 0mm 0mm 0mm 0mm,clip,scale=0.45]{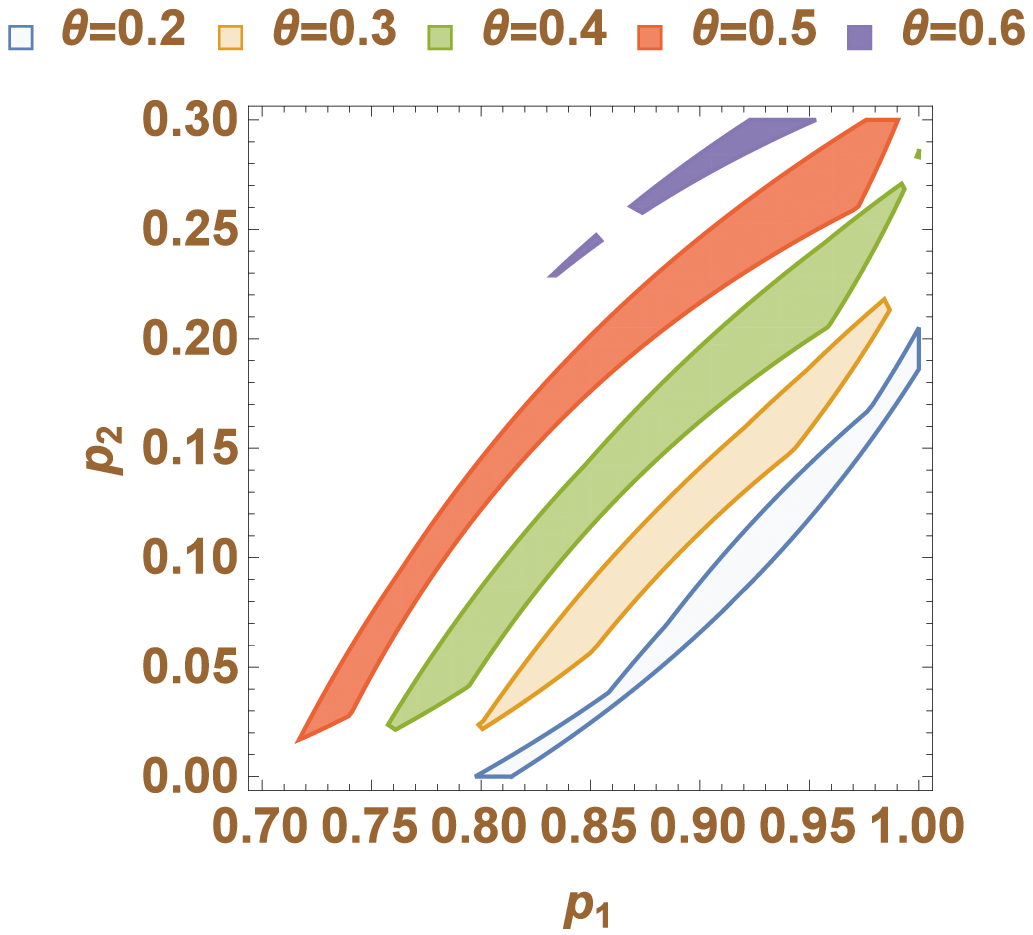}}&\subfloat []{\includegraphics[trim = 0mm 0mm 0mm 0mm,clip,scale=0.45]{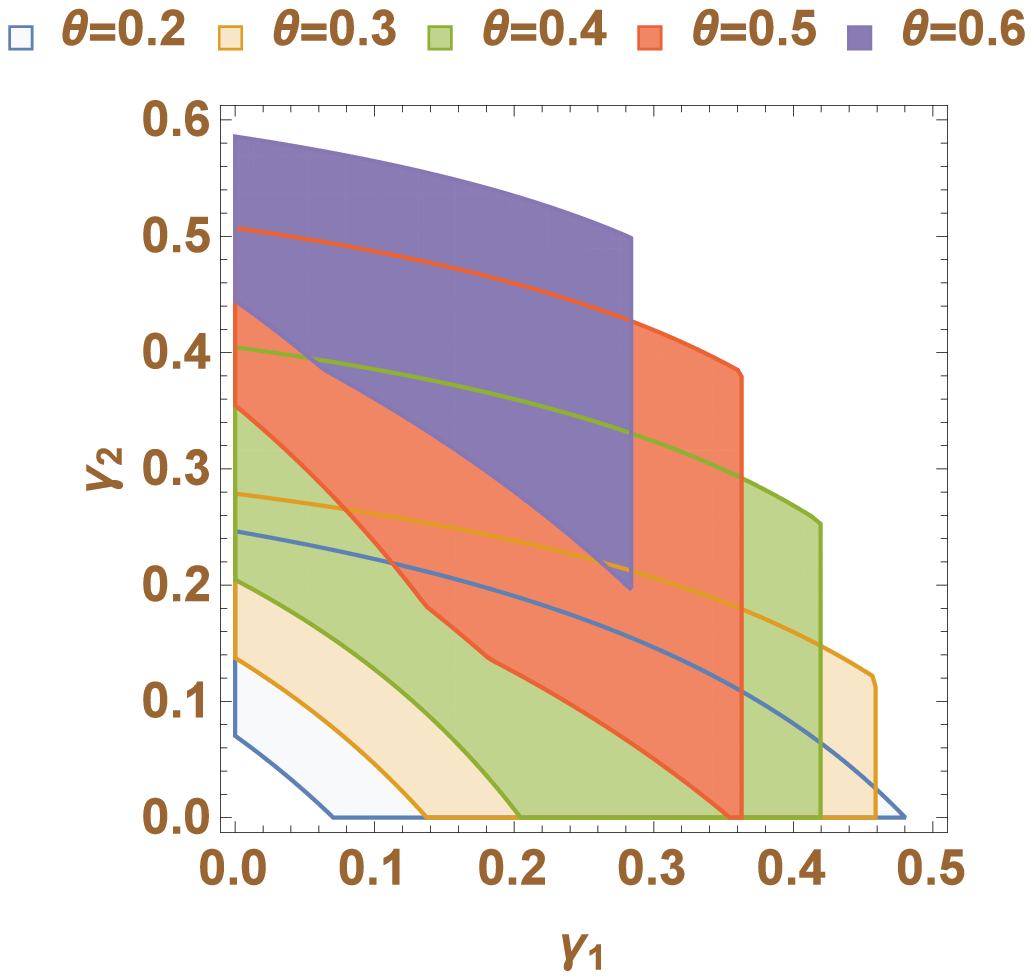}}\\
				\subfloat[]{\includegraphics[trim = 0mm 0mm 0mm 0mm,clip,scale=0.45]{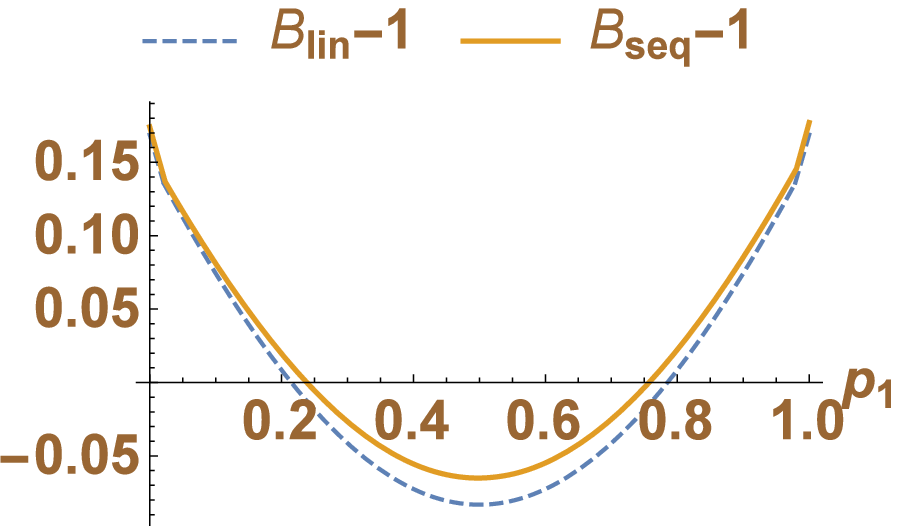}}&\subfloat[]{\includegraphics[trim = 0mm 0mm 0mm 0mm,clip,scale=0.45]{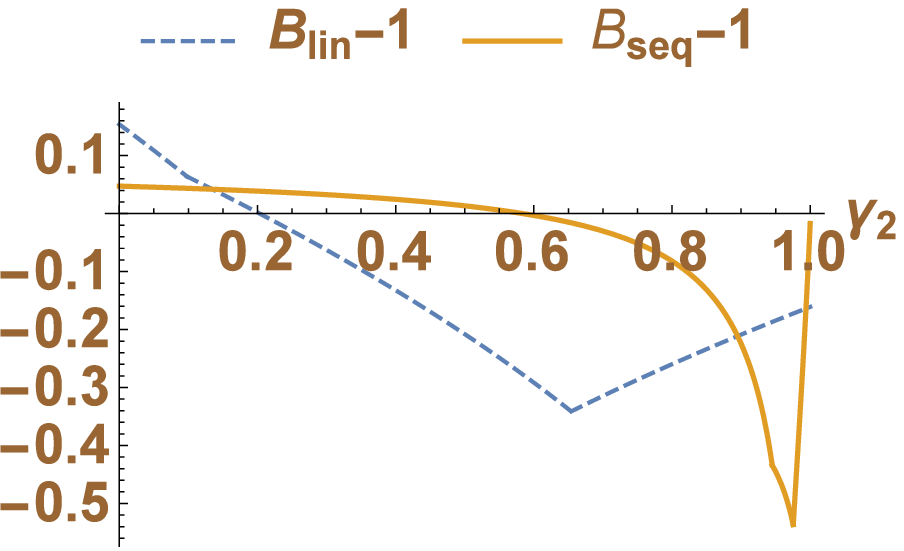}}\\
			\end{tabular}
			\caption{\emph{ All the subfigures in this figure point out the utility of applying suitable filters in perspective of increasing the bilocal inequality's resistance to noise when qubits are distributed across bit flip channel (all subfigures in left hand side panel) and amplitude damping channel (subfigures in right hand side panel). In the context of exploiting hidden non bilocality, in each of left and right panels: (i) the topmost one gives the possible region of both state and noise parameters, (ii) the middle one gives the region of noise parameters for different pure entangled states (Eq.(\ref{pure1})) , i.e., for different values of $\theta$ and (iii) the bottom one gives the range of one noise parameter for a fixed value of the other noise parameter and state parameter for which violation of bilocal inequality (Eq.(\ref{ineq})) is observed both with and without filtering operations.Details of the filtering parameters in each of the subfigures are as follow:$\epsilon_2^{(1)}$$=$$0.68$, $\epsilon_2^{(2)}$$=$$0.78$ in subfigures $a$ and $c$; $\epsilon_1$$=$$0.78$,$\epsilon_2^{(1)}$$=$$0.22$, $\epsilon_2^{(2)}$$=$$0.15,$ $\epsilon_3$$=$$0.73$ in subfigures $b$ and $d$; $\epsilon_2^{(1)}$$=$$0.98$, $\epsilon_2^{(2)}$$=$$0.79,$ $(\theta,p_2)$$=$$(0.58,0.15)$ in subfigure $e$ and $\epsilon_1$$=$$0.78$,$\epsilon_2^{(1)}$$=$$0.22$, $\epsilon_2^{(2)}$$=$$0.1,$ $\epsilon_3$$=$$0.79,$ $(\theta,\gamma_1)$$=$$(0.55,0.21)$ in subfigure $f$  
			}}
			\label{fignew2}
		\end{figure}
	\end{center}
	\section{Discussion}\label{conc}
	A sequential linear $n$-local network has been introduced in our present work. In the preparation stage of such a protocol, the parties are allowed to perform local filtering operations which constitute a specific form of stochastic local operations assisted with classical communication (SLOCC). Keeping analogy with hidden Bell nonlocality, non $n$-locality obtained in such protocols have been referred to as hidden non $n$-locality. Several instances of hidden non $n$-locality are demonstrated. This in turn points to the fact that filtering operations are significant in revealing hidden non $n$-locality. It is also observed that in some situations the sequential framework is more robust against noise than the usual network non $n$- locality.
	
	Interestingly, it is observed that hidden non $n$-locality can be observed even when one of the sources does not distribute entanglement. However same is not the case when one of the sources generates a product state. To this end, one may note that we have used a specific class of local filters which is however considered the most useful form of local filters in the standard Bell scenario \cite{hir1,hir2}. It will be interesting to characterize hidden non $n$-locality considering the general form of local filtering operations. Also, apart from applying local filters, considering other sequential measurement strategies to explore non $n$-locality can also be considered as a potential direction of future research. Besides, we have applied sequential measurement techniques in the linear $n$-local network scenario. It will be interesting to analyze similar techniques in the non-linear $n$-local networks.
	\section*{Acknowledgement}
	Tapaswini Patro would like to acknowledge the support from DST-Inspire fellowship No. DST/INSPIRE Fellowship/2019/IF190357. Nirman Ganguly acknowledges support from the project grant received under the SERB-MATRICS scheme vide file number MTR/2022/000101. We thank the anonymous referees for their insightful suggestions which has helped us to improve the overall quality of the manuscript. 
	\begin{widetext}
		\section*{Appendix}
		\begin{appendix}
			
			We first analyze the upper bound of $n$-local inequality (Eq.(\ref{ineq})) in the sequential $n$-local network.\\
			$\forall j$$=$$1,2,...,n,$ let source $\textbf{S}_i$ generate an arbitrary two qubit state $\rho_{j,j+1}$ (Eq.(\ref{st41})). In the preparation stage of the sequential network (sec.\ref{ress1}) $\textbf{A}_j,(j$$=$$1,2,...,n+1)$  applies local filter of the form given by Eqs.(\ref{fil5},\ref{fil6}).
			It may be noted that local filter $\mathfrak{F}_j$ (Eq.(\ref{fil6})) applied by each of $n-1$ intermediate parties $\textbf{A}_j(j$$=$$2,...,n)$ is of the form:
			\begin{eqnarray}\label{app1}
				\mathfrak{F}_j&=& \mathfrak{F}_j^{(1)}\otimes\mathfrak{F}_j^{(2)}\,\,\textmd{where}\\
				\mathfrak{F}_j^{(k)} &=& \epsilon_j^{(k)} |0\rangle\langle 0|+|1\rangle\langle 1|\,\, \textmd{for}\,k=1,2\,\textmd{and}\,j=2,3,...,n
			\end{eqnarray}
			As discussed in the main text (sec.\ref{ress1}), $n+1$-partite correlations generated at the end of the measurement stage are used to test the  $n$-local inequality (Eq.(\ref{ineq})).\\
			$n$-local inequality (Eq.(\ref{ineq})) is given by:
			\begin{equation}\label{app2}
				\frac{1}{2}\sum_{h=0}^1\sqrt{\textmd{Tr}[f_h(\textbf{M}_0,\textbf{M}_1,\textbf{N}_0,\textbf{N}_1)\rho_{\small{filtered}}]}\leq1
			\end{equation}
			In usual $n$-local network, Eq.(\ref{ineq}) is given by:
			\begin{eqnarray}\label{app2i}
				\frac{1}{2}\sum_{h=0}^1\sqrt{\textmd{Tr}[f_h(\textbf{M}_0,\textbf{M}_1,\textbf{N}_0,\textbf{N}_1)\rho_{\small{initial}}]}\leq 1\nonumber\\
				\frac{1}{2}\sum_{h=0}^1\sqrt{\textmd{Tr}[f_h(\textbf{M}_0,\textbf{M}_1,\textbf{N}_0,\textbf{N}_1)\otimes_{i=1}^n\rho_{i,i+1}]}\leq 1
			\end{eqnarray}
			As discussed in subsec.\ref{pre1}, upper bound (\textbf{B},say) of the above inequality (Eq.(\ref{app2i})), is given by \cite{bilo5}:
			\begin{equation}\label{app2ii}
				\textbf{B}=\sqrt{\Pi_{i=1}^nt_{i1}+\Pi_{i=1}^nt_{i2}},
			\end{equation}
			where $t_{i1},t_{i2}$ denoting largest two singular values of correlation tensor ($T_i$) of $\rho_{i,i+1}\,(i$$=$$1,2,...,n).$
			Now let us analyze the state $\rho_{filtered}$ used in above Eq.(\ref{app2}). As mentioned in sec.\ref{ress1}, $\rho_{\small{filtered}}$ (Eq.(\ref{fil7})) is given by:
			\begin{eqnarray}\label{app3}
				\rho_{filtered} &=& N(\otimes_{j=1}^{n+1}  \mathfrak{F}_{j})\rho_{\small{initial}}(\otimes_{j=1}^{n+1}  \mathfrak{F}_{j})^{\dagger},\,\,\textmd{where}\,N\, \textmd{is given by } Eq.(\ref{fil7})\nonumber\\
				&=&N\otimes_{j=1}^n \rho_{j,j+1}^{'}\,\textmd{where}\nonumber\\
				\rho_{1,2}^{'} &=& (\mathfrak{F}_{1}\otimes\mathfrak{F}_{2}^{(1)})\rho_{1,2}(\mathfrak{F}_{1}\otimes\mathfrak{F}_{2}^{(1)})^{\dagger}\nonumber\\
				\rho_{j,j+1}^{'} &=& (\mathfrak{F}_{j}^{(2)}\otimes\mathfrak{F}_{j+1}^{(1)})\rho_{j,j+1}(\mathfrak{F}_{j}^{(2)}\otimes\mathfrak{F}_{j+1}^{(1)})^{\dagger}\,\,\forall j=2,3,...,n-1\nonumber\\
				\rho_{n,n+1}^{'} &=& (\mathfrak{F}_{n}^{(2)}\otimes\mathfrak{F}_{n+1})\rho_{n,n+1}(\mathfrak{F}_{n}^{(2)}\otimes\mathfrak{F}_{n+1})^{\dagger}
			\end{eqnarray}
			It may be noted that $\forall j$$=$$1,2,...,n,\,\rho^{'}_{j,j+1}$ is unnormalized.  Let $\rho_{j,j+1}^{''}$ denote the normalized state corresponding to $\rho_{j,j+1}^{'}$:
			\begin{equation}\label{extra}
				\rho_{j,j+1}^{''}=N_{j,j+1}\rho_{j,j+1}^{'},
			\end{equation}
			where normalization factor $N_{j,j+1}$ is given by:
			\begin{eqnarray}\label{app4}
				N_{1,2}&=&\frac{1}{\textmd{Tr}[(\mathfrak{F}_{1}\otimes\mathfrak{F}_{2}^{(1)})\rho_{1,2}(\mathfrak{F}_{1}\otimes\mathfrak{F}_{2}^{(1)})^{\dagger}]}
				\nonumber\\
				N_{j,j+1}&=&\frac{1}{\textmd{Tr}[(\mathfrak{F}_{j}^{(2)}\otimes\mathfrak{F}_{j+1}^{(1)})\rho_{j,j+1}
					(\mathfrak{F}_{j}^{(2)}\otimes\mathfrak{F}_{j+1}^{(1)})^{\dagger}]}\,\,\forall j=2,3,...,n-1\nonumber\\
				N_{n,n+1}&=&Tr[ (\mathfrak{F}_{n}^{(2)}\otimes\mathfrak{F}_{n+1})\rho_{n,n+1}(\mathfrak{F}_{n}^{(2)}\otimes\mathfrak{F}_{n+1})^{\dagger}]
			\end{eqnarray}
			Now, Eq.(\ref{app3}) gives:
			\begin{eqnarray}\label{app5}
				\rho_{\small{filtered}}&=& N\otimes_{j=1}^n \rho_{j,j+1}^{'}\nonumber\\
				&=&N\otimes_{j=1}^n\frac{1}{N_{j,j+1}}(N_{j,j+1}\rho_{j,j+1}^{'})\nonumber\\
				&=&(\frac{N}{\otimes_{j=1}^n N_{j,j+1}})\otimes_{j=1}^n\rho_{j,j+1}^{''}\nonumber\\
				&=&\otimes_{j=1}^n\rho_{j,j+1}^{''}\,\textmd{\small{using}}\,Tr[\otimes_{i=1}^n R_i]=\Pi_{i=1}^n Tr[R_i],\,\textmd{\small{for any finite}}\, n
			\end{eqnarray}
			Using Eq.(\ref{app5}), Eq.(\ref{app2}) becomes:
			\begin{equation}\label{app6}
				\frac{1}{2}\sum_{h=0}^1\sqrt{\textmd{Tr}[f_h(\textbf{M}_0,\textbf{M}_1,\textbf{N}_0,\textbf{N}_1)\otimes_{j=1}^n\rho_{j,j+1}^{''}]}\leq1
			\end{equation}
			Comparison of Eq.(\ref{app2i}) with Eq.(\ref{app6}) points out that on maximizing over measurement parameters (used in measurement stage), upper bound ($\textbf{B}_{seq},$say) of above inequality and consequently that of $n$-local inequality (Eq.\ref{ineq}) in sequential $n$-local network is given by Eq.(\ref{app7})
			with $t^{''}_{j1},t^{''}_{j2}$ denoting largest two singular values of correlation tensor ($T^{''}_j$) of $\rho^{''}_{j,j+1}\,(j$$=$$1,2,...,n).$
			It may be noted that $ \textbf{B}_{seq}$ is a function of the Bloch parameters of $\rho_{j,j+1}\forall j$$=$$1,2,,...,n$ and the filtering parameters $\epsilon_1,\epsilon_{n+1},\epsilon_j^{(1)},\epsilon_j^{(2)}(j$$=$$2,3,...,n).$ Using Eq.(\ref{app7}), we next give the proof of the theorem.\\
			\textit{Proof of Theorem.1:} Let one of the $n$ sources generate product of two single qubit mixed states. W.L.O.G. let $\textbf{S}_1$ generate:
			\begin{equation}\label{app8}
				\rho_{1,2}=\frac{1}{4}(\sigma_0+\vec{m}.\vec{\sigma})\otimes (\sigma_0+\vec{n}.\vec{\sigma})\,\,\textmd{where}
			\end{equation}
			$\vec{m},\vec{n}$ are three dimensional vectors with length less than or equal to unity.
			Singular values of correlation tensor of $\rho_{1,2}$ are $(|\vec{m}||\vec{n}|,0,0).$ Magnitude of singular values of any correlation tensor is always less than unity\cite{hallsh} Hence when $\rho_{1,2}$ is used in the usual $n$-local network, $\textbf{B}$$\leq$$1.$ Consequently no violation of Eq.(\ref{ineq}) is obtained. Now, in a sequential $n$-local network, the correlation tensor of $\rho_{1,2}^{''}$ has only one non-zero singular value. So, from Eq.(\ref{app7}), we get $\textbf{B}_{seq}$$\leq$$ 1.$ Consequently violation of Eq.(\ref{ineq}) turns out to be impossible in sequential $n$-local network. Hence, if at least one of the sources generates product of two single-qubit mixed states, non $n$-locality cannot be detected in a sequential $n$-local network for any finite $n.$ This completes the proof of Theorem.1.$\blacksquare$\\
			\\
			\textit{Justification in support of conjecture made in subsec.\ref{ex3}:} As per condition, $n$-local inequality is not violated in usual $n$-local network. Hence, by Eq.(\ref{app2ii}),
			\begin{equation}\label{app9}
				\sqrt{\Pi_{j=1}^nt_{j1}+\Pi_{j=1}^nt_{j2}}\leq 1
			\end{equation}
			Let us focus on any one of the $n$ states $\rho_{j,j+1}(j$$=$$1,2,...,n).$ W.L.O.G.,let us consider $\rho_{1,2}.$ Local bloch vectors of $\rho_{1,2}$ are considered to be null. Singular values of $\rho_{1,2}^{''}$ turn out to be:
			\begin{eqnarray}
				t_{1,1}^{''} &=& \frac{\epsilon_1\epsilon_2^{(1)} t_{11}}{c_1}\nonumber\\
				t_{1,2}^{''} &=& \frac{\epsilon_1\epsilon_2^{(1)} t_{12}}{c_1}\\
				t_{1,3}^{''}&=& \frac{(1-\epsilon_1^2)(1-(\epsilon_2^{(1)})^2)+t_{13}(1+\epsilon_1^2)(1+(\epsilon_2^{(1)})^2)
				}{4c_1}\,\textmd{where}\,\nonumber\\
				c_1&=&t_{13}(1-\epsilon_1^2)(1-(\epsilon_2^{(1)})^2)+(1+\epsilon_1^2)(1+(\epsilon_2^{(1)})^2)
			\end{eqnarray}
			Singular values of $\rho_{j,j+1}^{''}(j$$=$$2,3,...,n)$ have analogous forms. For these forms of singular values, numerical maximization of Eq.(\ref{app7}), under the constraint that Eq.(\ref{app9}) holds, yields $1.$ Consequently Eq.(\ref{ineq}) is not violated in case none of $\rho_{j,j+1}$ has local Bloch vectors.
		\end{appendix}
	\end{widetext}
	
\end{document}